\documentclass[a4paper]{IEEEtran}
\usepackage{flushend}
\usepackage{bm}
\usepackage{mathrsfs} 
\usepackage[tbtags]{amsmath}
\usepackage{amssymb}
\usepackage{amsthm}
\usepackage{cite}
\usepackage{url}
\usepackage{hyperref}
\usepackage{cleveref}
\usepackage{float}
\usepackage{array}
\usepackage{graphicx}
\usepackage{epstopdf}
\usepackage[justification=centering]{caption}
\usepackage{subfigure}
\usepackage{diagbox}%\diagbox{\tabincell{c}{MAE and \\ RMSE}}{\tabincell{l}{Case}}

\usepackage{geometry}
\usepackage{makecell}%±í¸ñµ¥Ôª¸ñ»»ÐÐ
\usepackage{color}
\usepackage{multirow}

\theoremstyle{plain}%assumption,definition,remark,plain
\newtheorem{theorem}{Theorem}
\theoremstyle{definition}
\newtheorem{assumption}{Assumption}
\theoremstyle{definition}

\theoremstyle{definition}
\newtheorem{property}{Property}
\theoremstyle{definition}
\newtheorem{lemma}{Lemma}
\theoremstyle{remark}
\newtheorem{remark}{Remark}

\crefname{assumption}{Assumption}{Assumptions}
\crefname{equation}{}{}
\crefname{figure}{Fig.}{Figs.}
\crefname{table}{TABLE}{TABLEs}
\crefname{theorem}{Theorem}{Theorems}
\newcommand{\tabincell}[2]{\begin{tabular}{@{}#1@{}}#2\end{tabular}}  %??????

 \title{Adaptive Fractional-Order Sliding Mode Controller with Neural Network Compensator for an Ultrasonic Motor}
 
 \author{{Xiaolong~Chen, Wenyu~Liang,~\IEEEmembership{Member,~IEEE,} Han~Zhao, and~Abdullah~Al~Mamun,~\IEEEmembership{Senior Member,~IEEE}}% <-this % stops a space
 	\thanks{X. Chen, and H. Zhao are with the School of Mechanical Engineering, Hefei University of Technology, Hefei China 230009.}
 	\thanks{X. Chen, W. Liang, and A. Al Mamun are with the Department of Electrical and Computer Engineering, National University of Singapore, Singapore 117582.}
	 \thanks{X. Chen, and W. Liang are also with the Institute for Infocomm Research, A*STAR Singapore 138632.}
 \thanks{This work has been submitted to the IEEE for possible publication.
 	Copyright may be transferred without notice, after which this version may
 	no longer be accessible.}}
 
\date{}
\begin{document}
\maketitle

\begin{abstract}
Ultrasonic motors (USMs) are commonly used in aerospace, robotics, and medical devices, where fast and precise motion is needed. Remarkably, sliding mode controller (SMC) is an effective controller to achieve precision motion control of the USMs. To improve the tracking accuracy and lower the chattering in the SMC, the fractional-order calculus is introduced in the design of an adaptive SMC in this paper, namely, adaptive fractional-order SMC (AFOSMC), in which the bound of the uncertainty existing in the USMs is estimated by a designed adaptive law. Additionally, a short memory principle is employed to overcome the difficulty of implementing the fractional-order calculus on a practical system in real-time. Here, the short memory principle may increase the tracking errors because some information is lost during its operation. Thus, a compensator according to the framework of Bellman's optimal control theory is proposed so that the residual errors caused by the short memory principle can be attenuated. Lastly, experiments on a USM are conducted, which comparative results verify the performance of the designed controller.
\end{abstract}


\section{Introduction}
%\indent\indent 
\IEEEPARstart{U}{ltrasonic} motors (USMs) have excellent performances in fast response, high accuracy, quiet working, compact size \cite{tian2020review}, and thus they have extensively been applied in many industrial applications, such as aerospace, robotics, and medical devices, where fast and precise motion is needed. The USM is a typical nonlinear mechanical system. When such a system is in operation, the model uncertainties, hysteresis nonlinearity, inherent friction, and heat disturbance can all destroy the system's operating accuracy and thus the desired motions is failed to be achieved. Hence, it is crucial to design a precision motion controller with strong robustness for such a system.

Model-free controller \cite{zhang2019integrated,xu2017model,zhang2018model,Yik2019motion,saadatmand2020power} and  model-based controller \cite{liang2018contact,ma2021optimal,lau2019adaptive,chen2020optimal,chen2020novel} are two effective methods to achieve desired motions. The former one is implemented in a system without considering its accurate prior model information of dynamics or environment, which makes such a method be easily extended to carry out in most mechanical systems. In this method, the system stability is guaranteed by the tuning or training of feedback gains or some terms in the controller. However, for a specific complex nonlinear system in an uncertain or unstable environment, the stability is easy to collapse when the stability relies on the tuning method. The training method could make the system have a better adaptive capacity in an unstable environment, but it may need the use of high-performance computers to achieve satisfactory results through a large number of iterations. Inversely, for the model-based controller, the prior model information is considered in controller design, which could provide faster response and better performance than the model-free controller in controlling a specific complex nonlinear system \cite{zhang2020trajectory}. Therefore, the model-based controller design also becomes an active field in dealing with the trajectory tracking problem of complex nonlinear systems.

Sliding mode controller (SMC) with prior model information is one of the most widely used model-based controllers, which has strong robustness, fast convergence, and good trajectory tracking accuracy \cite{ding2020second,chin2018robust}. However, because there exists the presence of unmodeled dynamics, disturbances, and uncertainties in the system, a high-frequency phenomenon called chattering is always arisen in implementing SMC \cite{lee2007chattering}. The chattering effect could reduce the accuracy or damage the actuators, as pointed out in several works \cite{ma2017integrated,ma9374800,ma2021maglev}. Significatnly, the chattering is relevant to the disturbance \cite{song2014adaptive}. A disturbance could lead to significant changes of the position errors and velocity errors in the sliding surface function, especially for the latter one. Then, the signum function of SMC generates an obvious chattering because of the error changes. One popular way to alleviate this phenomenon is the boundary layer method in which a continuous function, such as a sigmoid function or a saturation function, is employed to substitute the signum function \cite{drakunov1992sliding,wang2020iterative}. However, this method can result in reduction of system robustness and increase in the tracking errors\cite{modiri2020adaptive}.

Another effective way to alleviate the chattering phenomenon is to utilize the fractional-order calculus during the design of SMC \cite{moezi2019optimal,fei2019experimental,sun2018practical}. By comparing with the classical integer-order calculus, the fractional-order calculus has the properties of long-range correlation and history dependence, and thus its result contains richer information. In recent years, many researchers have gradually tried to introduce the fractional-order calculus into the controller design so that higher controller flexibility can be achieved. To deal with the fractional-order differential of a mismatched disturbance, a fractional-order disturbance observer is designed in \cite{wang2018fractional}. To enhance the precise tracking response and robust control performance of a linear permanent magnet synchronous motor (PMSM) system, a fractional-order sliding mode controller (FOSMC) is developed in \cite{chen2019precision}. Furthermore, a FOSMC strategy is applied to control the speed of a PMSM and reduce the chattering effect \cite{zaihidee2019application}.  

Because of the fractional-order calculus's advantages, all the proposed methods mentioned above have a less chattering effect, lower overshoot, and faster response speed. However, to maintain the heredity and long-range correlation, the calculation process requires the computer to have a huge memory to store historical data \cite{wei2017note,pano2019fpga,petravs2019fractional}. As a result, there is a challenge of the fractional-order calculus to carry out on a practical system in real-time for a long period. In \cite{podlubny1999fractional}, a short memory principle is proposed to provide a solution for this problem, which only takes into account the behavior in the recent past. Here, although the principle provides a compromise solution for experiments, the ignored information in implementing the principle may render to an increase in the tracking errors. Taking all the above matters into consideration, an adaptive fractional-order sliding mode controller (AFOSMC) and a compensator according to the framework of Bellman’s optimal control theory are proposed in this paper for the motion control of USMs. Significantly, the short memory principle is employed in implementing the designed controller.

The main contributions of this paper are summarized below. Firstly, the nonlinear dynamical model with the time-varying bounded but unknown uncertainty of the USM is built. Secondly, a fraction-order sliding surface function is designed, and then an AFOSMC is proposed to ensure that the USM can run smoothly and accurately. Here, an adaptive law is proposed to adapt the bound of the uncertainty, and the short memory principle is employed to solve the numerical solution of the fractional derivative during the implementation of the AFOSMC. Furthermore, the stability of the AFOSMC is analyzed by Lyapunov theory. Thirdly, a compensator is designed in this paper to attenuate the residual errors caused by the short memory principle. Significantly, this compensator is based on the framework of Bellman's optimal control theory, and a neural network updated by gradient descent algorithm is developed to estimate the optimal cost function. At the same time, the Lyapunov theory verifies that the dynamics of the neural network is ultimately uniformly bounded.

The rest of this paper is structured as follows. At first, the dynamical model of a class of USMs with the uncertainty is introduced in Section II. Next, Section III shows the concepts of the fraction-order calculus and details the design of the proposed controller. Following that in Section IV, the experiments are carried out on a USM in real-time to show the practical performance of the designed controller under both continuous and discontinuous references. Finally, Section V gives the conclusions.

\section{Dynamical model of USM}
%\indent\indent 
The dynamical model with the uncertainty of a USM system can be given by the following equation:
\begin{equation}\label{dynamical_model}
\begin{split}
&m(\varrho(t))\ddot{q}(t)+b(\varrho(t))\dot{q}(t)+c(\varrho(t))q(t)=g(\varrho(t))\mu(t),
\end{split}
\end{equation}
where $t\in \mathbb{R}$ is the time, $q \in \mathbb{R}$ is the generalized coordinate, $\dot{q} \in \mathbb{R}$ and $\ddot{q} \in \mathbb{R}$ are the velocity and acceleration, respectively, $\varrho \in \bm{\sum}\subset \mathbb{R}^p$ is the time-varying uncertain parameter (that is possibly fast). Here, $\bm{\sum}\subset \mathbb{R}^p$ is compact but unknown, which represents the possible bound of $\varrho$. $m \in \mathbb{R}$, $b\in \mathbb{R}$, $c\in \mathbb{R}$ are the mass of the system, damping coefficient, stiffness, respectively, $g\in \mathbb{R}$ is the control gain, $\mu \in  \mathbb{R}$ is the control input to the system.

\begin{assumption}\label{assumption_m}
For each $t\in  \mathbb{R}$, $\varrho \in \bm{\sum}$,  $m(\varrho(t))>0$.
\end{assumption}

Considering the uncertain parameter $\varrho$ in the system \cref{dynamical_model}, the parameters $m$, $b$, $c$, $g$ can be decomposed as
\begin{equation}\label{decomposed_m}
	\begin{split}
		m(\varrho):=\bar{m}+\Delta m (\varrho),
	\end{split}
\end{equation}
\begin{equation}\label{decomposed_b}
	\begin{split}
		b(\varrho):=\bar{b}+\Delta b (\varrho),
	\end{split}
\end{equation}\begin{equation}\label{decomposed_c}
\begin{split}
	c(\varrho):=\bar{c}+\Delta c (\varrho),
\end{split}
\end{equation}\begin{equation}\label{decomposed_g}
\begin{split}
	g(\varrho):=\bar{g}+\Delta g (\varrho),
\end{split}
\end{equation}
where $\bar{m}$, $\bar{b}$, $\bar{c}$ and $\bar{g}$ denote the nominal parts and $\bar{m}$ is positive, while $\Delta m$, $\Delta b$, $\Delta c$ and $\Delta g$ denote the uncertain parts. Let $h:=\bar{m}^{-1}$, $\Delta h(\varrho):=m^{-1} (\varrho)-\bar{m}^{-1}$, and then we have $m^{-1} (\varrho)=h+\Delta h(\varrho)$. 

The model-based controller is designed based on a dynamical model like \cref{dynamical_model}. Even if system identification can obtain a very accurate dynamical model, the uncertainty and disturbance still exist. For example, the hysteresis nonlinearity, inherent friction, and heat disturbance always exist in the operation of USMs, which are (possibly fast) time-varying and difficult to identify. Thus, it is essential to consider the uncertainty in model-based controller design to improve controller robustness and accuracy.

\section{Controller design}
\subsection{Fractional-order sliding mode controller design}
%\indent\indent
A fundamental operator of the fractional calculus can be denoted by:
\begin{equation}\label{fundamental_fractional_calculus}
_{t_0}\mathscr{D}^\alpha_t\cong \mathscr{D}^\alpha= \left\{ \begin{aligned}
&\frac{\mathscr{D}^\alpha}{dt^\alpha}, &\text{if}\quad  \alpha> 0,\\
&1,  &\text{if} \quad \alpha= 0,\\
&\int_{t_0}^t (d\tau)^{-\alpha}, &\text{if} \quad \alpha< 0,\\
\end{aligned}
\right.
\end{equation}
where $\alpha \in \mathbb{R}$ is the fractional-order, $t_0$ and $t$ are the two limits of the operation.
\begin{remark}
When the fractional-order $\alpha$ is positive, the derivative operator is carried out. In contrast, when the fractional-order $\alpha$ is negative, the integral operator is carried out.
\end{remark}

There are several definitions of the general fractional operator, and for a continuous function $\eta (t)$, two most commonly used operators are Riemann--Liouville (RL) definition and Gr$\ddot{\mathrm{u}}$nwald--Letnikov (GL) definition.

The RL fractional integral operator of order $\alpha$ of function $\eta(t)$ is defined by \cite{luo2012fractional}:
\begin{equation}\label{RL_fractional_integral}
_{t_0}I_t^\alpha \eta(t)=\frac{1}{\Gamma(\alpha)}\int_{t_0}^t \eta(\tau)(t-\tau)^{\alpha-1}d\tau,
\end{equation}
where $\alpha>0$, $\Gamma(\alpha)=\int_0^\infty  \tau^{\alpha-1} \exp(-\tau) d\tau$.
%where $\alpha>0$, $\Gamma(\cdot)$ is the Euler's Gamma function and it is defined as $\Gamma(\alpha)=\int_0^\infty  \tau^{\alpha-1} \exp(-\tau) d\tau$.

The RL fractional derivative operator of order $\alpha$ of function $\eta(t)$ is defined by %\cite{luo2012fractional}
\begin{equation}\label{RL_fractional_derivative}
^{RL}_{t_0}\mathscr{D}_t^\alpha \eta(t)=\frac{d^{n_f}}{dt^{n_f}} [I^{n_f-\alpha}\eta(t)],
\end{equation}
where $n_f$ is the smallest integer greater than or equal to $\alpha$. 

The GL fractional derivative operator of order $\alpha$ of function $\eta(t)$ is defined by %\cite{luo2012fractional}
\begin{equation}\label{GL_fractional_derivative}
^{GL}_{t_0}\mathscr{D}_t^\alpha \eta(t) = \lim_{\varsigma \rightarrow 0^+} \frac{_{t_0}\tilde{\Delta}_t^\alpha \eta(t) }{\varsigma^\alpha},
\end{equation}
and
\begin{equation}\label{GL_fractional_Delta}
_{t_0}\tilde{\Delta}_t^\alpha \eta(t) =  \sum_{j=0}^{[\frac{t-t_0}{\varsigma}]} \tilde{c}_j \eta(t-j\varsigma),
\end{equation}
where $[\frac{t-t_0}{\varsigma}]$ represents the integer part of $\frac{t-t_0}{\varsigma}$, $ \tilde{c}_j=(1-\frac{1+\alpha}{j})\tilde{c}_{j-1}$ is the binomial coefficient and $\tilde{c}_0=1$, $\varsigma$ is the time-step.

\begin{remark}
Note that the GL fractional derivative operator coincides with the RL fractional derivative operator if the considered function $\eta(t)$ is sufficiently smooth. Therefore, it is allowed to use the RL fractional derivative operator during problem formulation and then turn to the GL fractional derivative operator for obtaining the numerical solution \cite{podlubny1999fractional}. 
\end{remark}

\begin{property}
Considering two continuous functions $\eta_1(t)$ and $\eta_2(t)$ and two arbitrary scalars $\zeta_1$ and $\zeta_2$,  the following property holds for $\forall~\alpha>0$ \cite{luo2012fractional,podlubny1999fractional},
\begin{equation}\label{RL_property_1}
	\mathscr{D}^{\alpha}[\zeta_1 \eta_1(t)+\zeta_2\eta_2(t)]=\zeta_1\mathscr{D}^{\alpha} \eta_1(t)+\zeta_2\mathscr{D}^{\alpha} \eta_2(t).
\end{equation}
\end{property}

\begin{property}
Considering a continuous function $\eta(t)$,  the following property holds for $\forall~\alpha_1>0$, $\alpha_2>0$, \cite{luo2012fractional,podlubny1999fractional}
\begin{equation}\label{RL_property_2_1}
\mathscr{D}^{\alpha_1} [\mathscr{D}^{ - \alpha_2}\eta(t)]=\mathscr{D}^{\alpha_1 - \alpha_2} \eta(t).
\end{equation}
Specially, when $\alpha_1=\alpha_2$, we have
\begin{equation}\label{RL_property_2_2}
\mathscr{D}^{\alpha_1 - \alpha_2} \eta(t)=\mathscr{D}^0  \eta(t) =\eta(t).
\end{equation}
\end{property}

\begin{property}
Considering a continuous function $\eta(t)$,  the following property holds  for $\forall~m\in \mathbb{N}^*$, and $\alpha>0$,  \cite{luo2012fractional,podlubny1999fractional} 
	\begin{equation}\label{RL_property_3_1}
		\frac{d^m}{dt^m} \mathscr{D}_t^\alpha \eta(t) = \mathscr{D}_t^{m+\alpha} \eta(t).
	\end{equation}
Specially, when $m-1\leq\alpha<m$, we have
\begin{equation}\label{RL_property_3_2}
	\frac{d^m}{dt^m} \mathscr{D}_t^{-(m-\alpha)} \eta(t) = \mathscr{D}_t^{\alpha} \eta(t).
\end{equation}
\end{property}

\begin{remark}
The numerical solution of the fractional derivative relies on all historical data, which makes the fractional derivatives possess to be long memory characteristic. The characteristic is the essential differences between the fractional calculus and the integral one. In the real world, the implementation of a fractional derivative is performed by using high-performance computers. Therefore, its solution may take up much hardware resources and increase the burden on computers.
\end{remark}

To ensure computational efficiency and real-time performance, especially in feedback control, the short memory principle is implemented in solving the numerical solution of a fractional derivative.

\begin{lemma}
(Short memory principle \cite{podlubny1999fractional}) For a continuous function $\eta(t)$, the following approximation with fixed memory length $L$  is given by
\begin{equation}\label{Short_memory}
_{t_0}\mathscr{D}^\alpha_t \eta(t) \approx_{t-L}\mathscr{D}_{t}^\alpha \eta(t)=\frac{1}{\varsigma^\alpha}  \sum_{j=0}^{[\frac{L}{\varsigma}]} \tilde{c}_j \eta(t-j\varsigma),
\end{equation}
where $0<\alpha<1$, $t>t_0+L$, $L>0$ is the memory length, $[\frac{L}{\varsigma}]$ represents the integer part of $\frac{L}{\varsigma}$.
\end{lemma}

This principle considers the experience of $\eta(t)$ only in the "recent past", which helps to reduce the burden of the processor in solving numerical solution of a fractional derivative. However, this implementation pays a fine in the form of some inaccuracy. In detail, if $\eta(t) \leq M$ for $t_0 \leq t\leq T$, we have %\cite{podlubny1999fractional}
\begin{equation}\label{Delta_d}
	\begin{split}
		\Delta d(t)&:=|_{t_0}\mathscr{D}^\alpha_t \eta(t) - _{t-L}\mathscr{D}_{t}^\alpha \eta(t)| \leq \frac{ML^{-\alpha}}{|\Gamma(1-\alpha)|}.
	\end{split}
\end{equation}

We can find that the value of the right-hand side (RHS) of the above inequality is determined by the memory length $L$. Therefore, the memory length can adjust the required accuracy $\varpi$. Then, for $t_0+L \leq t\leq T$, we have
\begin{equation}
	\begin{split}
	&\Delta d(t)\leq \varpi, \ if \ L \geq (\frac{M}{\varpi |\Gamma(1-\alpha)|})^\frac{1}{\alpha },
	\end{split}
\end{equation}
which is also the choice principle of the memory length $L$ based on the system accuracy $\varpi$.

Now, let $e(t):=q(t)-q_r(t)$ (here, $q_r$ is the reference trajectory) and then we introduce a fractional-order sliding surface function $s(t)$ as follows:
\begin{equation}\label{sliding_surface}
s(t)=\lambda e(t)+  \mathscr{D}^{1+\alpha} e(t) ,
\end{equation}
where $\lambda$ and $\alpha$ are positive constants, and $0<\alpha<1$.

\begin{assumption}\label{Control_assumption_Pi}
	(i) There exist  an unknown constant vector $\bm{\beta}\in(0,\infty)^k$, and a known function $\Pi(\cdot):(0,\infty)^k\times\mathbb{R} \times \mathbb{R} \times \mathbb{R}\rightarrow \mathbb{R}_+$, such that for $\forall~(q,\dot{q},t)\in \mathbb{R}^n \times \mathbb{R}^n \times \mathbb{R}$,
	\begin{equation}\label{Control_assumption_Pi1}
		\begin{aligned}
			&\max\limits[| \mathscr{D}^{\alpha}[k_s \mathrm{sgn} (s(t))-h\Delta b(\varrho)\dot{q}(t)-h\Delta c(\varrho)q(t)\\
			&\quad+h\Delta g(\varrho) \mu_{s}(t)-\Delta h(\varrho) b(\varrho)\dot{q}(t)-\Delta h(\varrho)c(\varrho) q(t)\\
			&\quad+\Delta h(\varrho)g(\varrho)\mu_{s}(t) ] |]\leq\Pi(\bm{\beta},q,\dot{q},t),
		\end{aligned}
	\end{equation}
where $k_s>0$ is a constant, $\mathrm{sgn}(\cdot)$ is the signum function.
\noindent
	(ii) For each set of $(\bm{\beta},q,\dot{q},t)$, we can linearly factorize the function $\Pi(\bm{\beta},q,\dot{q},t)$ with respect to $\bm{\beta}$, and therefore a function $\bar{\bm{\Pi}}(\cdot):\mathbb{R} \times \mathbb{R} \times \mathbb{R}\rightarrow \mathbb{R}_+^k$ exists which makes
	\begin{equation}\label{Control_assumption_Pi2}
		\Pi(\bm{\beta},q,\dot{q},t)=\bm{\beta}^T\bar{\bm{\Pi}}(q,\dot{q},t).
	\end{equation}
\end{assumption}
 
Based on the dynamical model \cref{dynamical_model} and the fractional-order sliding surface function \cref{sliding_surface}, we propose the following controller (namely, AFOSMC):
\begin{equation}\label{fractional_smc}
\begin{split}
\mu_{s}(t)=&\bar{g}^{-1}[\bar{b}\dot{q}(t)+\bar{c}q(t) +\bar{m}\ddot{q}_{r}(t)-\lambda \bar{m} \mathscr{D}^{-\alpha} \dot{e}(t)\\
&-k_p \bar{m} \mathscr{D}^{-\alpha} s(t)-k_s\bar{m}  \mathrm{sgn} (s(t))\\
&-\bar{m} \mathscr{D}^{-\alpha}\frac{s\Pi^2(\hat{\bm{\beta}},q,\dot{q},t)}{\epsilon(t)}],
\end{split}
\end{equation}
with $\textcolor{black}{	
			\dot{\epsilon}(t)=-\bar{l} \epsilon (t).
}$
$k_p$, $\bar{l}$ and $\epsilon(t_0)$ are positive constants, $\hat{\bm{\beta}}$ is the estimated value of $\bm{\beta}$, and $\hat{\alpha}_{\bar{x}}(t_0)>0$ (where $\hat{\alpha}_{\bar{x}}$ is the $\bar{x}$-th component of the vector $\hat{\bm{\beta}}$, $\bar{x}=1,2,\dots,k$)

The vector $\hat{\bm{\beta}}$ is updated by the adaptive law shown below
\begin{equation}\label{hat_beta}
\dot{\hat{\bm{\beta}}}(t)=k_1 |s(t)| \bar{\bm{\Pi}}(q,\dot{q},t),
\end{equation}
where $k_1$ is positive constant.

\begin{theorem}
Consider a mechanical system \cref{dynamical_model} that is subjected to the reference trajectory $q_r$, and suppose that \cref{assumption_m,Control_assumption_Pi} are satisfied. Let $\delta(t):=[s(t) \ (\hat{\bm{\beta}}(t)-\bm{\beta})^T\ 0]\in\mathbb{R}^{1+k+1}$. The controller \cref{fractional_smc} renders the performance as shown below:
		\begin{itemize}
		\item[(i)] Uniform stability: there is a $\gamma>0$ for each $\chi>0$, such that $\|\delta(t)\|< \chi$ for $\forall~t\geq t_0$ if $\delta(\cdot)$ is any solution with $\|\delta(t_0)\|<\gamma$.
		\item[(ii)] Convergence to zero: for any given $\delta(\cdot)$,
			\begin{equation}\label{}
			\begin{split}
\lim\limits_{t\rightarrow \infty} s (t)=0.
			\end{split}
		\end{equation}
	\end{itemize}
\end{theorem}

\begin{proof}
A Lyapunov function is selected as %\cite{Khalil2002Nonlinear}
	\begin{equation}\label{proof_1}
		V_1(s,\hat{\bm{\beta}}-\bm{\beta},\epsilon,t)=\frac{1}{2}s^2+\frac{1}{2}k_1^{-1}(\hat{\bm{\beta}}-\bm{\beta})^T(\hat{\bm{\beta}}-\bm{\beta})+\frac{ 1 }{4}\bar{l}^{-1}\epsilon.
	\end{equation}
	
 According to \cref{sliding_surface}, we have the time derivative of $V_1$ as
	\begin{equation}\label{proof_2}
			\dot{V}_1=s\dot{s}+k_1^{-1}(\hat{\bm{\beta}}-\bm{\beta})^T\dot{\hat{\bm{\beta}}}+\frac{ 1 }{4}\bar{l}^{-1}\dot{\epsilon}.
	\end{equation}
For the first term on the RHS of \cref{proof_2}, and recalling \cref{controller_mu}, we have
	\begin{equation}\label{proof_3}
		\begin{split}
			s\dot{s}&=s( \lambda\dot{e}+ \mathscr{D}^{\alpha}\ddot{e})\\
			&=s[\lambda\dot{e}+ \mathscr{D}^{\alpha}(\frac{g}{m}\mu_s-\frac{b}{m}\dot{q}-\frac{c}{m}q-\ddot{q}_r)]\\
			&=s[\lambda\dot{e}+ \mathscr{D}^{\alpha}(-h\bar{b}\dot{q}-h\bar{c}q+h\bar{g}\mu_{s}-\ddot{q}_r)\\
			&\quad+\mathscr{D}^{\alpha}(-h\Delta b\dot{q}-h\Delta cq+h\Delta g \mu_s-\Delta h b\dot{q}\\
			&\quad -\Delta hc q+\Delta hg\mu_s )].\\
		\end{split}
	\end{equation}
	
	Substitute \cref{fractional_smc} into \cref{proof_3}, we have
	\begin{align}\label{proof_4}
			s\dot{s}=&-k_p  s^2-\frac{s^2\Pi^2(\hat{\bm{\beta}},q,\dot{q},t)}{\epsilon}+s\mathscr{D}^{\alpha}[k_s \mathrm{sgn} (s)-h\Delta b\dot{q}\nonumber\\
		&-h\Delta cq+h\Delta g \mu_s-\Delta h b\dot{q}-\Delta hc q+\Delta hg\mu_s].
	\end{align}
	
	With \cref{Control_assumption_Pi}(i), we have
	\begin{equation}\label{proof_5}
		\begin{split}
	s\dot{s}&\leq -k_p  s^2-\frac{s^2\Pi^2(\hat{\bm{\beta}},q,\dot{q},t)}{\epsilon}+|s|\Pi(\bm{\beta},q,\dot{q},t).
		\end{split}
	\end{equation}

With \cref{Control_assumption_Pi}(ii), we have
	\begin{equation}\label{proof_6}
		\begin{split}
		    &- \frac{s^2\Pi^2(\hat{\bm{\beta}},q,\dot{q},t)}{\epsilon}+|s| \Pi(\bm{\beta},q,\dot{q},t)\\
			& =[|s| \Pi(\hat{\bm{\beta}},\bm{\chi},\dot{\bm{\chi}},t)- s^2\frac{ \Pi^2(\hat{\bm{\beta}},q,\dot{q},t)}{\varepsilon}]\\
			&\quad+[|s|\Pi(\bm{\beta},q,\dot{q},t)-|s| \Pi(\hat{\bm{\beta}},q,\dot{q},t)]\\
			&\leq \frac{\varepsilon}{4}+|s|(\bm{\beta}-\hat{\bm{\beta}})^T \bar{\bm{\Pi}}(q,\dot{q},t).\\
		\end{split} 
	\end{equation}

Substitute \cref{proof_5,proof_6} into \cref{proof_2}, we have
	\begin{equation}\label{proof_7}
	\begin{split}
		\dot{V}_1&\leq \frac{\varepsilon}{4}+|s|(\bm{\beta}-\hat{\bm{\beta}})^T \bar{\bm{\Pi}}(q,\dot{q},t)+\frac{ 1 }{4}\bar{l}^{-1}\dot{\epsilon}\\
		&\quad +k_1^{-1}(\hat{\bm{\beta}}-\bm{\beta})^T\dot{\hat{\bm{\beta}}}\\
	& =\frac{\varepsilon}{4}+|s|(\bm{\beta}-\hat{\bm{\beta}})^T \bar{\bm{\Pi}}(q,\dot{q},t)+\frac{ 1 }{4}\bar{l}^{-1}\dot{\epsilon}\\
	&\quad +k_1^{-1}(\hat{\bm{\beta}}-\bm{\beta})^T\dot{\hat{\bm{\beta}}}.
	\end{split}
\end{equation}

Based on $\textcolor{black}{\dot{\epsilon}(t)}$ and substitute \cref{hat_beta} into \cref{proof_2}, we get
	\begin{align}\label{proof_8}
		\dot{V}_1&\leq -k_p  s^2+\frac{\varepsilon}{4}-|s|(\hat{\bm{\beta}}-\bm{\beta})^T\bar{\bm{\Pi}}(q,\dot{q},t)+\frac{ 1 }{4}\bar{l}^{-1}(-\bar{l} \epsilon)\nonumber\\
		&\quad+k_1^{-1}(\hat{\bm{\beta}}-\bm{\beta})^T(k_1 |s| \bar{\bm{\Pi}}(q,\dot{q},t)) =-k_p  s^2.
\end{align}

Because the Lyapunov derivative is non-positive according to \cref{proof_8}, we have the uniform stability. Moreover, from \cref{proof_8}, we have $k_p  s^2 \leq -\dot{V}_1$, and then we have 
\begin{equation}\label{proof_9}
	\begin{split}
		&\lim\limits_{t\rightarrow \infty}\int_0^t k_p  s^2(\Gamma) d \Gamma \\
		& \leq 	V_1(s,\hat{\bm{\beta}}-\bm{\beta},\epsilon,0)- V(s,\hat{\bm{\beta}}-\bm{\beta},\epsilon,\infty)\\
		&\leq V_1(s,\hat{\bm{\beta}}-\bm{\beta},\epsilon,0).
	\end{split}
\end{equation}

Therefore, according to Barbalat's lemma \cite{hou2010new}, we can have the following conclusion: $\lim\limits_{t\rightarrow \infty} 
s^2 =0$. Then, we can conclude that the convergence of $s$ to $0$ as $t\rightarrow \infty$.	
\end{proof}
\subsection{Compensator design}
In this subsection, a compensator is designed to deal with the residual errors that are caused by the short-memory principle. Firstly, there is a vector $\bm{I}(t)=[\lambda_1 \int e(t)dt\ \lambda_2 e(t) \  \lambda_3 \dot{e}(t)]^T$, where $\lambda_1$, $\lambda_2$ and $\lambda_3$ are positive constants. Then, we have its time derivative:
\begin{equation}\label{compensator_model}
\begin{split}
\dot{\bm{I}}(t)=\bm{f}_c(t)+\bm{g}_c(t)\mu_c(t) ,
\end{split}
\end{equation}
where \\ $\bm{f}_c(t)=\begin{bmatrix}
\lambda_1 e &
\lambda _2\dot{e} &
\lambda _3 (-\frac{b}{m}\dot{q}-\frac{c}{m}+\bar{g}\Delta h\mu_c+\Delta g h \mu_c)
\end{bmatrix}^T$, $\bm{g}_c(t)=\begin{bmatrix}0& 0& \lambda_3 \bar{g}h\end{bmatrix}^T$, $\mu_c$ is a compensation controller.

To design the feedback-based approximate optimal controller $u_c$, we firstly define a cost function:
\begin{equation}
\begin{split}
J(\bm{I})=\int_{t_s}^t \bm{I}^T(\tau)\bm{Q}\bm{I}(\tau) + \mu_{c}^T(\tau) R \mu_{c}(\tau) d\tau,
\end{split}
\end{equation}
where $\bm{Q}$ is symmetric positive-definite matrices with appropriate dimension, $R$ is a positive constant.

Then the optimal controller $u^*_{c}$ should satisfy the following optimal cost function:
\begin{equation}
	\begin{split}
		J^*(\bm{I})=&\min_{\mu_c(\tau)|\tau\in R_{\geq t_s}}\int_{t_s}^\infty \bm{I}^T(\tau)\bm{Q}\bm{I}(\tau)+  \mu_{c}^T(\tau) R  \mu_{c}(\tau) d\tau.
	\end{split}
\end{equation}

For \cref{compensator_model}, we have the following Hamiltonian function:
\begin{equation}\label{Hamiltonian_function}
	\begin{split}
		H(\bm{I},\dot{\bm{I}},\mu_c,\nabla J)&=\bm{I}^T(t)\bm{Q}\bm{I}(t) +  \mu_{c}^T(t) R  \mu_{c}(t)\\
		&\quad+\nabla J^{T}(\bm{I}) \dot{\bm{I}}(t),
	\end{split}
\end{equation}
where $\nabla J=\frac{\partial J}{\partial \bm{I}}$.

Based on the Bellman's optimal control theory, the optimal value function $J^*(\bm{I})$ satisfies the Hamilton-Jacobi-Bellman equation:
\begin{equation}
	\begin{split}
		\min_{\mu_c(\tau)|\tau\in R_{\geq t_s}} H(\bm{I},\dot{\bm{I}}, \mu_{c}^*,\nabla J^*)=0.
	\end{split}
\end{equation}

By solving the equation $\frac{\partial H(\bm{I},\dot{\bm{I}}, \mu_{c}^*,\nabla J^*)}{\partial  \mu_{c}^*(t)}=0$, we have the optimal compensator, which can be expressed as 
\begin{equation}
	\begin{split}
		\mu_{c}^*=-\frac{1}{2}\bm{R}^{-1}\bm{g}_c^T\nabla J^*.
	\end{split}
\end{equation}

With the help of a neural network, the optimal cost function can also be expressed as
\begin{equation}\label{optimal_value_function}
	\begin{split}
		J^*(\bm{I})=\bm{W}^{*T} \bm{\sigma}(\bm{I}) +\bm{\varepsilon}_c(\bm{I}),
	\end{split}
\end{equation}
where $\bm{W}^*\in \mathbb{R}^{n_c}$ is the bounded optimal wight value of the neural network, $\bm{\sigma}\in \mathbb{R}^{n_c}$ is the bounded activation function, and $\bm{\varepsilon}_c\in \mathbb{R}^{n_c}$ is the bounded reconstruction error.

Substitute \cref{optimal_value_function} into \cref{Hamiltonian_function}, we have
\begin{equation}\label{neural_Hamiltonian_function}
	\begin{split}
		H(\bm{I},\dot{\bm{I}},\mu_c,\nabla J^*)&=\bm{I}^T(t)\bm{Q}\bm{I}(t) +  \mu_{c}^T(t) R  \mu_{c}(t)\\
		&\quad+ \bm{W}^{*T} \nabla\bm{\sigma}\dot{\bm{I}}(t) +\nabla \bm{\varepsilon}_c \dot{\bm{I}}(t),
	\end{split}
\end{equation}
where $\nabla \bm{\sigma}=\frac{\partial \bm{\sigma}}{\partial \bm{I}}$, and $\nabla \bm{\varepsilon}_c=\frac{\partial \bm{\varepsilon}_c}{\partial \bm{I}}$.

To get $J^*(\bm{I})$ by function \cref{optimal_value_function} is always Utopian, but we can have its estimated value by
\begin{equation}
\begin{split}
\hat{J}(\bm{I})=\hat{\bm{W}}^T \bm{\sigma}(\bm{I}),
\end{split}
\end{equation}
where $\hat{\bm{W}}\in \mathbb{R}^{n_c}$ is the estimate of the optimal wight $\bm{W}^*$, and the bounded error between $\hat{\bm{W}}$ and $\bm{W}^*$ is noted as $\tilde{\bm{W}}=\hat{\bm{W}}-\bm{W}^*$.

Approximating \cref{neural_Hamiltonian_function} yields
\begin{equation}\label{estimated_Hamiltonian_function}
	\begin{split}
		\hat{H}(\bm{I},\dot{\bm{I}}, \mu_{c},\nabla \hat{J})&=\bm{I}^T(t)\bm{Q}\bm{I}(t) +  \mu_{c}^T(t) R  \mu_{c}(t)\\
		&\quad+\nabla \hat{J} ^T(\bm{I}) \dot{\bm{I}}(t),\\
	\end{split}
\end{equation}
where $\nabla \hat{J}=\frac{\partial \hat{J}}{\partial \bm{I}}=\nabla \bm{\sigma}^T\hat{\bm{W}}$.

Now, we have the approximate value of $ \mu_{c}^*$ as
\begin{equation}\label{approximate_optimal_compensator}
 \mu_{c}=-\frac{1}{2}\bm{R}^{-1}\bm{g}_c^T(t)\nabla \hat{J}=-\frac{1}{2}\bm{R}^{-1}\bm{g}_c^T(t)\nabla \bm{\sigma}^T \hat{\bm{W}}.
\end{equation}

We have the objective function $E=\frac{1}{2} \kappa \hat{H}^2$. By using the gradient descent algorithm $-\frac{\partial E}{\partial \hat{\bm{W}}}$, the $\hat{\bm{W}}$ is updated based on the following adaptive law
\begin{equation}\label{adaptive_law}
\begin{split}
\dot{\hat{\bm{W}}}=-   \kappa\hat{H} \nabla \bm{\sigma}\dot{\bm{I}},
\end{split}
\end{equation}
where $\kappa$ is positive constants.

\begin{theorem}
Consider a mechanical system \cref{dynamical_model} subjected to the reference trajectory $q_r$, then the  adaptive law \cref{adaptive_law} can render the dynamics of $\tilde{\bm{W}}$ to be ultimately uniformly bounded.
\end{theorem}

\begin{proof}
 A Lyapunov function is selected as %\cite{Khalil2002Nonlinear}
\begin{equation}\label{proof_21}
	\begin{split}
		V_2(\tilde{\bm{W}},t)&=\frac{1}{2}\kappa^{-1}\tilde{\bm{W}}^T\tilde{\bm{W}}.\\
	\end{split}
\end{equation}

Then, we have the time derivative of $V_2$ as
\begin{equation}\label{proof_22}
	\begin{split}
		\dot{V}_2&=\kappa^{-1}\tilde{\bm{W}}^T\dot{\hat{\bm{W}}}.\\
	\end{split}
\end{equation}

Then, by \cref{adaptive_law}, we have
\begin{equation}\label{proof_23}
	\begin{split}
		\dot{V}_2&=-  \tilde{\bm{W}}^T \nabla \bm{\sigma}\dot{\bm{I}}\hat{H}.\\
	\end{split}
\end{equation}

Let $\nu_x=H-\nabla \bm{\varepsilon}_c \dot{\bm{I}}(t)$, and then we have

\begin{equation}\label{proof_24}
	\begin{split}
		\dot{V}_2&=-  \tilde{\bm{W}}^T\nabla \bm{\sigma}\dot{\bm{I}}(\nu_x+\tilde{\bm{W}}^T\nabla \bm{\sigma}\dot{\bm{I}}) \\
		&\leq - \frac{1}{2}\|\tilde{\bm{W}}^T\nabla \bm{\sigma}\dot{\bm{I}}\|^2+\frac{1}{2}\nu_x ^2. \\
	\end{split}
\end{equation}

By assuming that a positive constant  $\theta_{max}$ always exists such that $\|\nabla \bm{\sigma}\dot{\bm{I}}\|\leq \theta_{max}$. Then, we have $\dot{V}_2 \leq 0$ as long as $\tilde{\bm{W}}$ lies outside the compact set $\Omega=\{\tilde{\bm{W}}:\|\tilde{\bm{W}}\|\leq \|\frac{\nu_x}{\theta_{max}}\|\}$. Thus, the dynamics of $\tilde{\bm{W}}$ can be ensured to be ultimately uniformly bounded with above condition. 
\end{proof}

Finally, for the dynamical model \cref{dynamical_model}, we have
\begin{equation}\label{controller_mu}
\mu(t)= \mu_s(t)+\mu_c(t).                   
 \end{equation}

Also, by choosing the Lyapunov function $V=V_1+V_2$, we can conclude that the AFOSMC and compensator can render the system to be ultimately uniformly bounded.

\section{Experiments and Results}
%\indent\indent 
\cref{experimental_platform} shows the experimental platform of a USM, which consists of four parts: a computer, a dSPACE control card (DS1104), the USM, and a motor drive. The dSPACE control card is installed in the computer to carry out the controller with a sampling time of 1 $\mathrm{ms}$. The model number of the USM is PI-M663 Compact Linear Positioning Stage manufactured by Physik Instrumente. Therein, a linear encoder for position measurement is embedded in this motor with a resolution of 0.1 $\mathrm{\mu m}$. The model number of the motor drive is C-185.161 PI Line Motion Controller. Moreover, the weight of the parts mounted on the motor can be considered as a kind of system uncertainty.

\begin{figure}[H]
	\centering
	\includegraphics[scale=0.4,trim=200 80 200 80,clip]{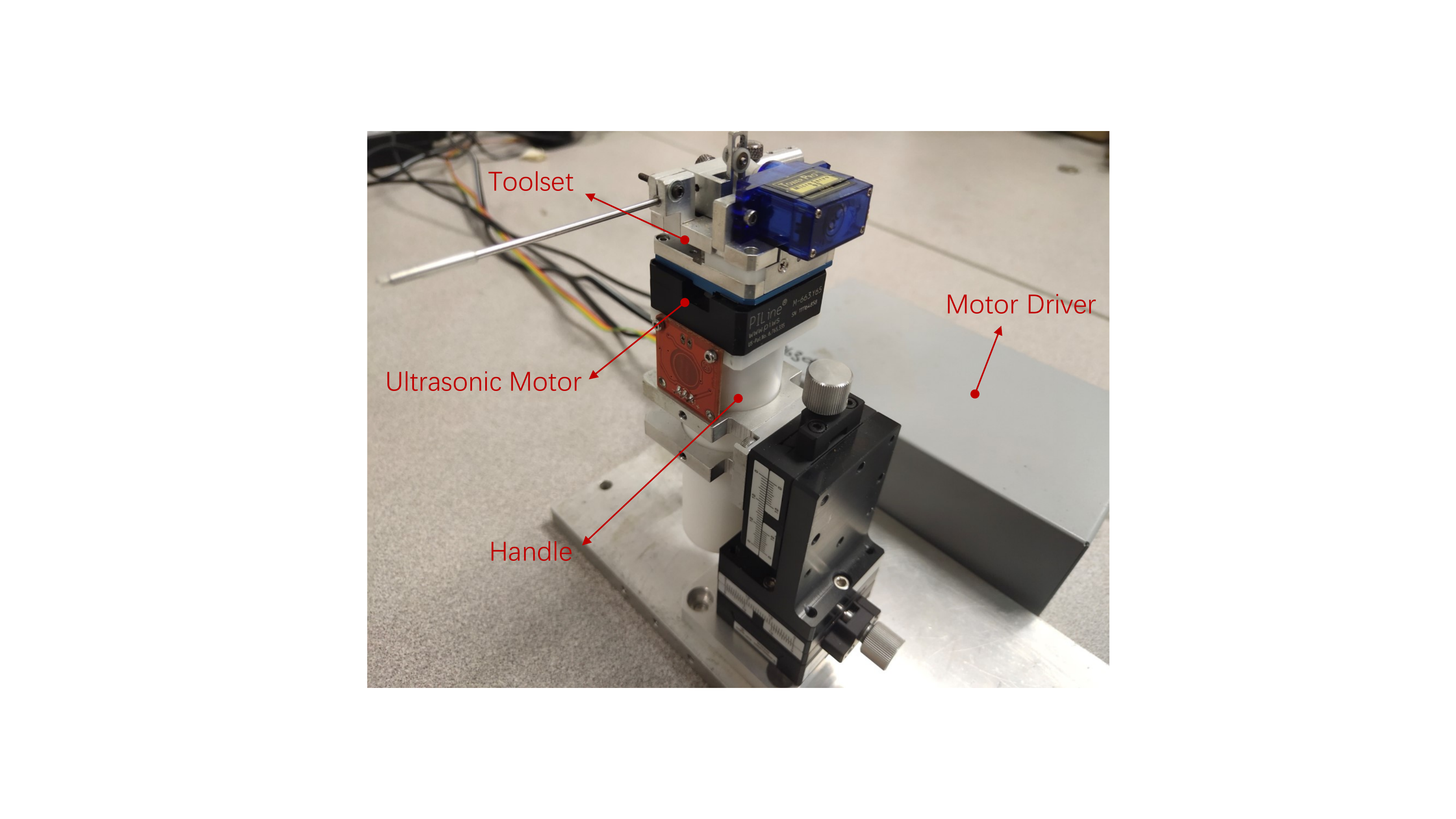}
	\caption{Experimental platform of a USM.\label{experimental_platform}}
\end{figure}
\subsection{Results and discussions}
%\indent\indent
Through system identification, we obtain the nominal system parameters as $\bar{m}=1$, $\bar{b}=248.4$, $\bar{c}=202$, $\bar{g}=4940$. To forcefully make an evaluation on the performance of the proposed controller, we employ a proportional-integral-derivative (PID) controller and an SMC. A disturbance observer (DOB) $\hat{d}(t)$ in \cite{lau2019adaptive} is added to further improve the property of the two controllers for comparison, and its bandwidth is chosen as $500$. Therefore, we have following three case in implementing the experiments:

Case 1: our designed controller

We can choose $\Pi(\bm{\beta},q,\dot{q},t)$ as (note that the motion requirements \cref{dynamical_model} are with respect to positions)
\begin{equation}
	\begin{split}
		&\Pi(\bm{\beta},q,\dot{q},t)=\beta_1| q|^2+ \beta_2 | q|+\beta_3\\
		& = \left[
		\begin{array}{ccc}
			\beta_1 & \beta_2 & \beta_3 
		\end{array}
		\right]
		\left[\begin{array}{c}
			 | q|^2\\
			 | q|\\
			1
		\end{array}\right],\\
	\end{split}
\end{equation}
where $\beta_1>0,\beta_2>0,\beta_3>0$ are constant parameters but unknown. It is worth noting that the following alternatives can also satisfy \cref{Control_assumption_Pi}:
\begin{equation}\label{equation pi for robot}
	\begin{split}
		&\beta_1| q|^2+ \beta_2 | q|+\beta_3  \leq \beta( | q|+1)^2 =:\beta\tilde{\Pi}(q,\dot{q},t),
	\end{split}
\end{equation}
where $\beta=\max\{\beta_1,\frac{\beta_2}{2},\beta_3\}$, $\tilde{\Pi}=( | q|+1)^2$. Also, the $\beta$ is an unknown constant that is estimated by the adaptive law \cref{hat_beta}.

Here, we choose the controller parameters as $\lambda=1$, $k_p=6000$, $k_s=5000$, $\bar{l}=0.01$, $k_1=0.1$, $\bm{Q}=\mathrm{I}_{3\times3}$, $R=494$, $\kappa=10^{-5}$, $\lambda_1=10$, $\lambda_2=1$, $\lambda_3=0.1$ $\bm{\sigma}(I)=[\frac{1}{1+ \exp(-I(1))}\ \frac{1}{1+ \exp(-I(2))}\  \frac{1}{1+ \exp(-I(3))}]^T$, the memory length $L=50$, and the initial values of $\hat{\alpha}$ and $\epsilon$ are chosen as $\hat{\alpha}(0)=0.1$ and $\epsilon(0)=1$, respectively.

Case 2: PID controller with DOB
\begin{equation}
		\mu_{pid}(t)=-k_{x1} e(t)-k_{x2} \int edt -k_{x3}\dot{e}(t)-\hat{d}(t),
\end{equation}
where the control parameters are choose as $k_{x1}=37$, $k_{x2}=0.1$, $k_{x3}=160$.

Case 3: SMC with DOB
\begin{align}
\mu_{smc}(t)&= \frac{\bar{m}}{\bar{g}}[\frac{\bar{b}}{\bar{m}}\dot{q}+\frac{\bar{c}}{\bar{m}}q+\ddot{q}_r- \lambda_{y} \dot{e}-k_{y1}s_{y}\nonumber\\
&\quad-k_{y2} \mathrm{sat}(\frac{s_{y}}{\sigma_{y}})-\frac{1}{\bar{m}}\hat{d}(t)],
\end{align}
where $s_{y}=\dot{e}+\lambda_{y}e$, $\mathrm{sat}(\cdot)$ is the saturation function, the control parameters are choose as $\lambda_{y} =50$, $k_{y1} =500$, $k_{y2}=300$, $\sigma_{y}=0.1$.

\subsubsection{Experimental results of sinusoidal tracking}
Sinusoidal waves are a typical continuous waveform in testing the tracking performances of control methods. We set the reference trajectory as $q_r(t)=\sin(2 \pi f t)$, where $f$ is the frequency. There are three frequencies of $f=1, 5,10$ Hz in this experiment to check the tracking performance of controllers from low to high frequency. 

\cref{x_sin_1Hz} shows the trajectory tracking results and errors at 1 Hz sinusoidal wave using different controllers. \cref{out_sin_1Hz} shows the control inputs at 1 Hz sinusoidal wave using different controllers. At 1 Hz, the nonlinear inherent friction can be one significant interference. As can be seen, the tracking errors of the PID controller with DOB are the worst among the three controllers because its ability to deal with nonlinear systems is weak. The peak value of the control input of SMC with DOB is obviously higher than that of the designed controller, but the tracking errors are worse than that of the designed controller. This result shows that the designed controller can achieve good system performance with less control effort/energy in the case of low frequency. 

\begin{figure}[!t]
	\centering
	\includegraphics[scale=0.18,trim=40 10 50 0,clip]{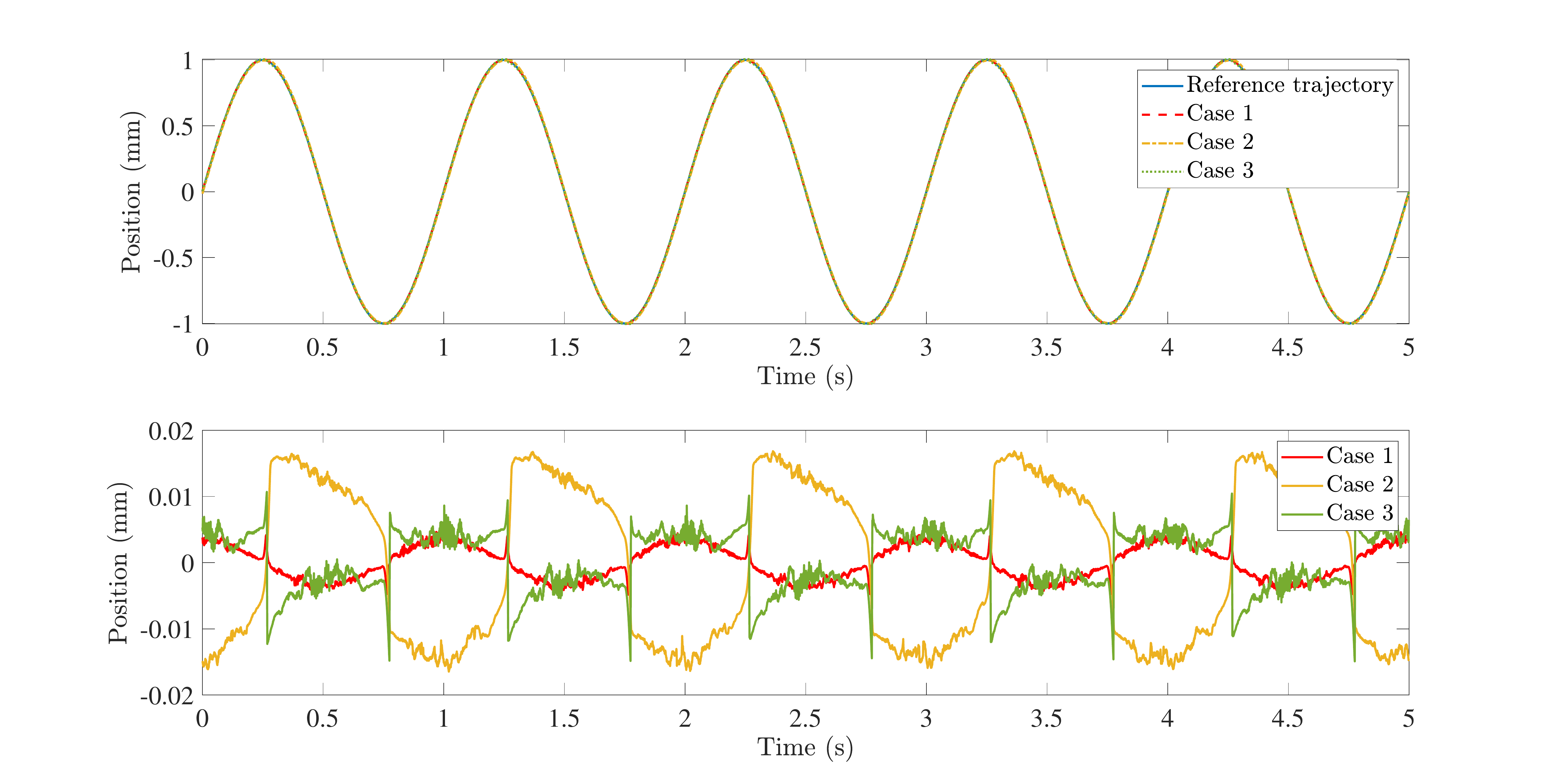}
	\caption{Trajectory tracking results at 1 Hz sinusoidal wave using different controllers (\textbf{Top.} Position tracking performance; \textbf{Bottom.} Position tracking errors).\label{x_sin_1Hz}}\vspace{-0.5cm}
\end{figure}
\begin{figure}[!b]
	\centering\vspace{-0.5cm}
	\includegraphics[scale=0.18,trim=40 10 50 0,clip]{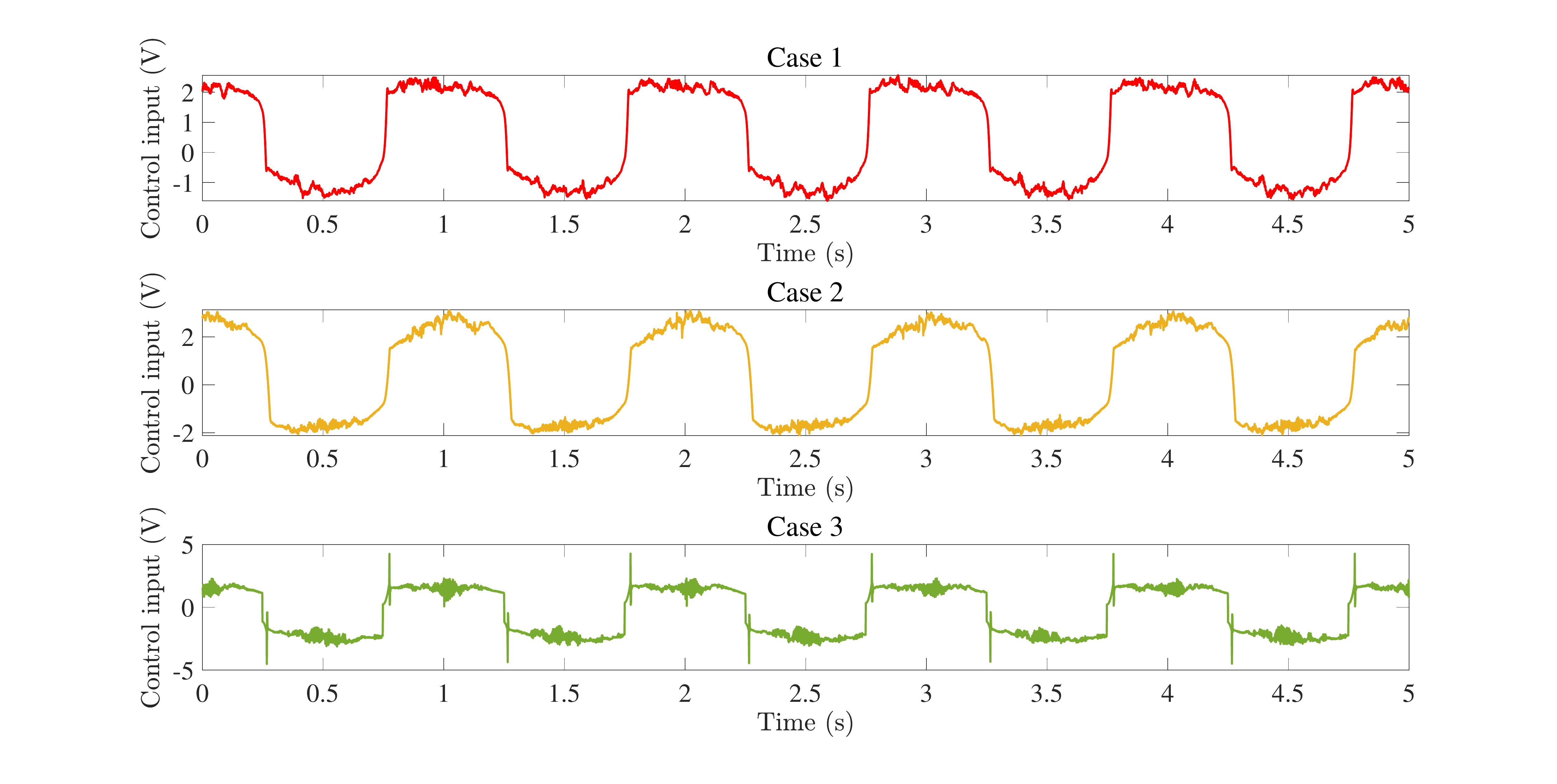}
	\caption{Comparison of control inputs at 1 Hz sinusoidal wave among different controllers.\label{out_sin_1Hz}}
\end{figure}

\cref{x_sin_5Hz,x_sin_10Hz} show the trajectory tracking results and errors at 5 Hz and 10 Hz sinusoidal waves using different controllers, respectively. \cref{out_sin_5Hz,out_sin_10Hz} show the control inputs at 5 Hz and 10 Hz sinusoidal waves using different controllers, respectively. At high frequency, heat disturbance becomes another significant interference that affects the system performance. As the increase in frequency, the influence of heat disturbance will be greater. As can be seen from these figures, at either 5 Hz or 10 Hz, the tracking error of the designed controller is the smallest among the three controllers. Furthermore, comparing with SMC with DOB, the chattering of control inputs of the designed controller is smaller, which means that the designed controller is more friendly to the hardware system. 

\begin{figure}[!t]
	\centering
	\includegraphics[scale=0.18,trim=40 10 50 0,clip]{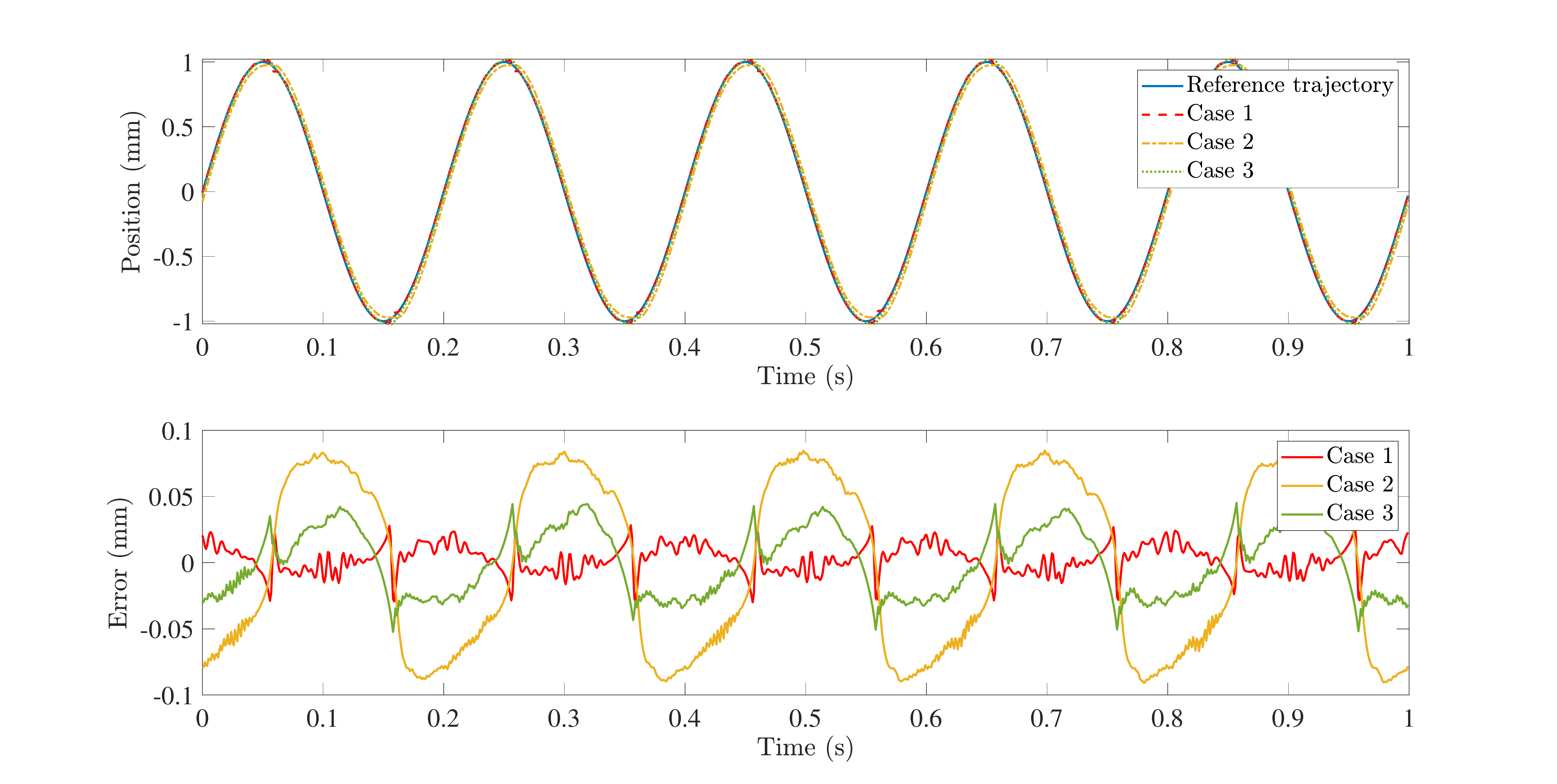}
	\caption{Trajectory tracking results at 5 Hz sinusoidal wave using different controllers (\textbf{Top.} Position tracking performance; \textbf{Bottom.} Position tracking errors).\label{x_sin_5Hz}}\vspace{-0.5cm}
\end{figure}
\begin{figure}[!t]
	\centering
	\includegraphics[scale=0.18,trim=40 10 50 0,clip]{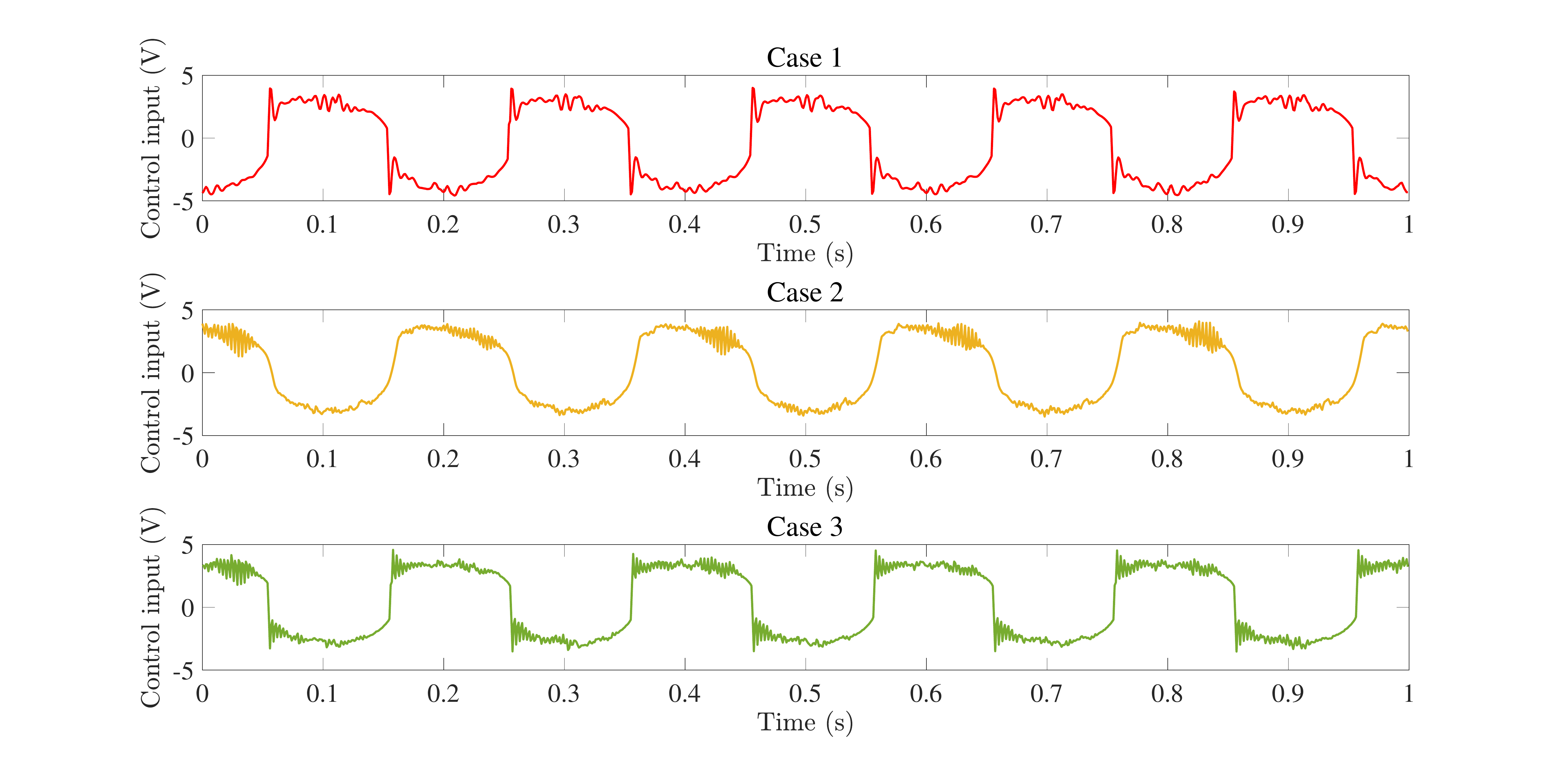}
	\caption{Comparison of control inputs at 5 Hz sinusoidal wave among different controllers.\label{out_sin_5Hz}}\vspace{-0.5cm}
\end{figure}

\begin{figure}[!t]
	\centering
	\includegraphics[scale=0.18,trim=40 10 50 0,clip]{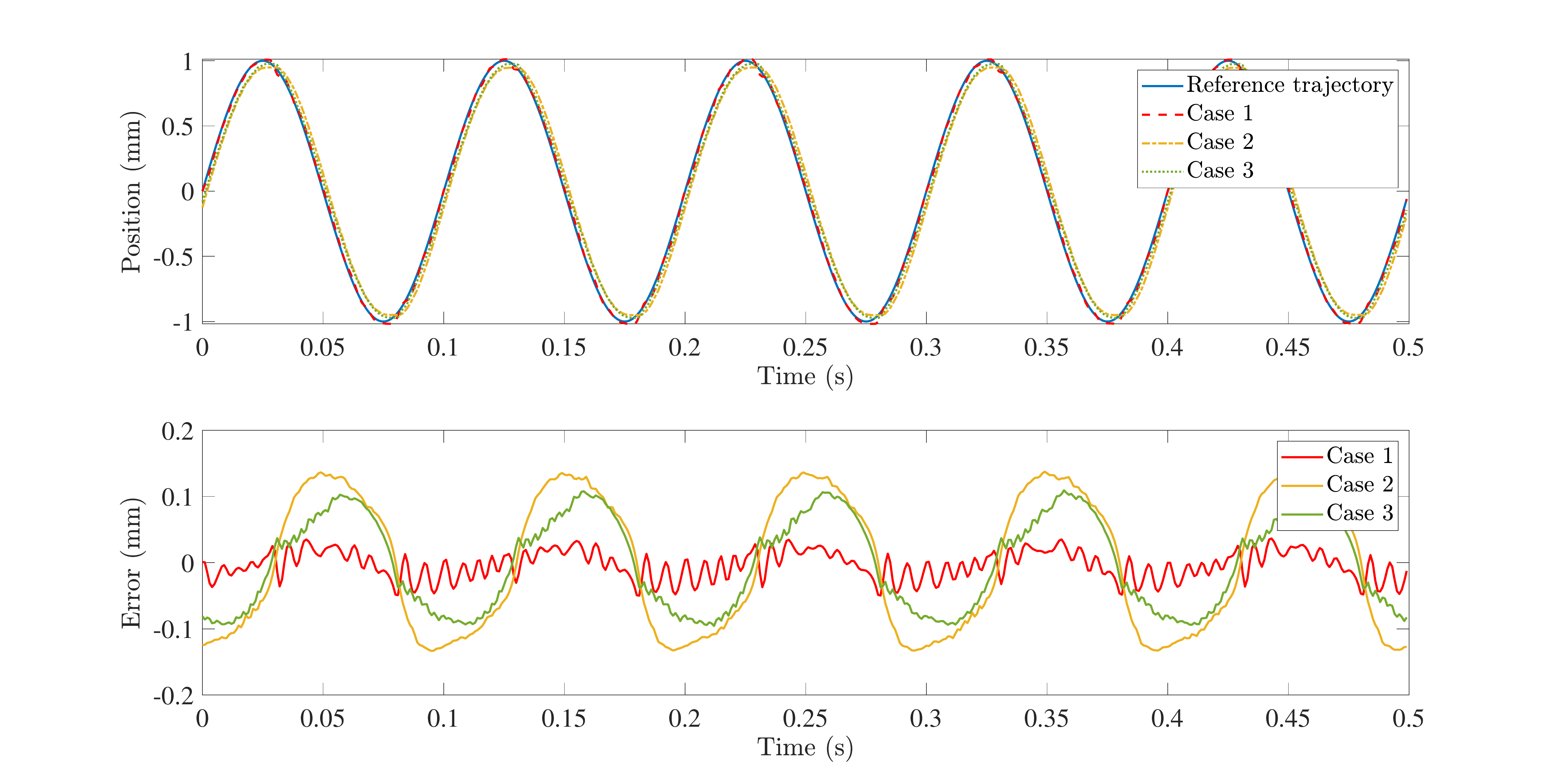}
	\caption{Trajectory tracking results at 10 Hz sinusoidal wave using different controllers (\textbf{Top.} Position tracking performance; \textbf{Bottom.} Position tracking errors).\label{x_sin_10Hz}}\vspace{-0.5cm}
\end{figure}
\begin{figure}[!t]
	\centering
	\includegraphics[scale=0.18,trim=40 10 50 0,clip]{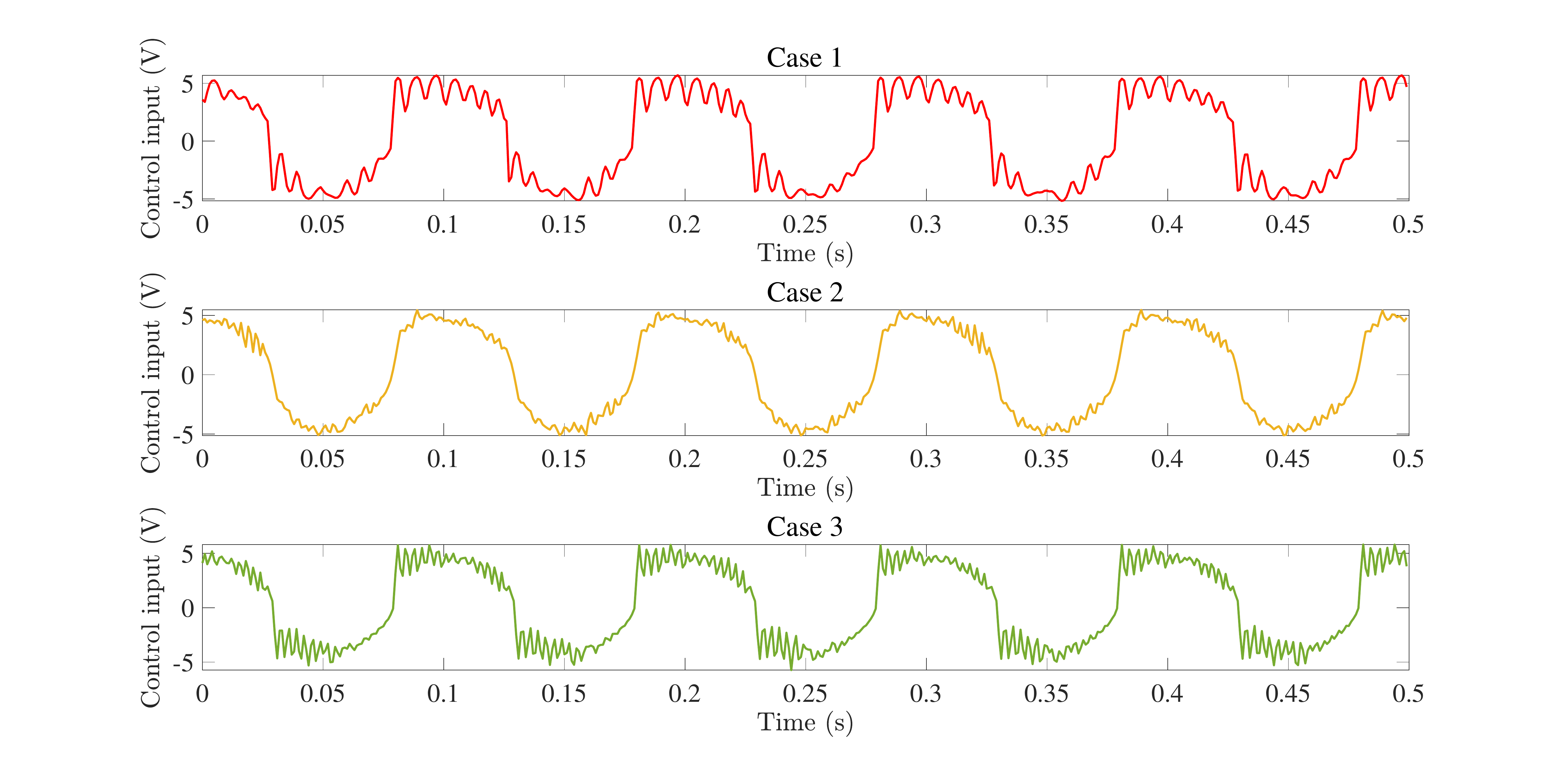}
	\caption{Comparison of control inputs at 10 Hz sinusoidal wave among different controllers.\label{out_sin_10Hz}}\vspace{-0.5cm}
\end{figure}

\cref{error1_table} exhibits the comparisons of the maximum absolute errors (MAEs) and root-mean-square errors (RMSEs) among the three controllers, where Case 4 is the results of our previous work reported in \cite{lau2019adaptive} (a controller on basis of an adaptive sliding mode enhanced disturbance observer  for the same experimental platform). In terms of the RMSE, the designed controller can achieve smaller tracking errors which are reduced by 21.88\% at 1 Hz and 21.07\% at 10 Hz. By comparing all cases in \cref{error1_table}, it is easy to conclude that the designed controller achieves better tracking performance than other cases in tracking continuous waves.

\begin{table}[!t]
	\centering
	\caption{Tracking errors of different cases for sinusoidal waves with different frequencies.\label{error1_table}}
	\newsavebox{\firsttablebox}
	\begin{lrbox}{\firsttablebox}
		\begin{tabular}{lcccc}
			\hline
			\hline
			\tabincell{l}{Cases }& \tabincell{l}{MAE and RMSE} & 1Hz& 5 Hz & 10Hz  \\
			\hline
			\multirow{2}*{\tabincell{l}{Case 1}}   &  MAE (mm)    &   0.0046       &0.0110  & 0.0365 \\
			&RMSE (mm)      &   0.0025     &0.0028&0.0251 \\
			\hline
			\multirow{2}*{\tabincell{l}{Case 2}}  &  MAE   (mm)       &  0.0168  &0.0847&0.1375\\
			&RMSE  (mm)      &   0.0121    &0.0647 &0.1010\\
			\hline
			\multirow{2}*{\tabincell{l}{Case 3}}    &  MAE  (mm)       &   0.0107    &0.0452 &0.1093\\
			&RMSE  (mm)      &    0.0046      &0.0252 &0.0695\\
			\hline
			\multirow{2}*{\tabincell{l}{Case 4}}  &  MAE   (mm)       &  0.0053   &0.0318&0.0961 \\
			&RMSE  (mm)      &   0.0032      &0.0187&0.0318 \\
			\hline
			\hline
		\end{tabular}
	\end{lrbox}
	\scalebox{1}{\usebox{\firsttablebox}}
\end{table}

\subsubsection{Experimental results of triangular tracking}
To further examine the tracking performance of the designed controller for discontinuous waveform, the triangular wave is employed in this experiment. Here, $\ddot{q}_r$ at the changing points is set to be zero. 

\begin{figure}[!b]
	\centering\vspace{-0.5cm}
	\includegraphics[scale=0.18,trim=40 10 50 0,clip]{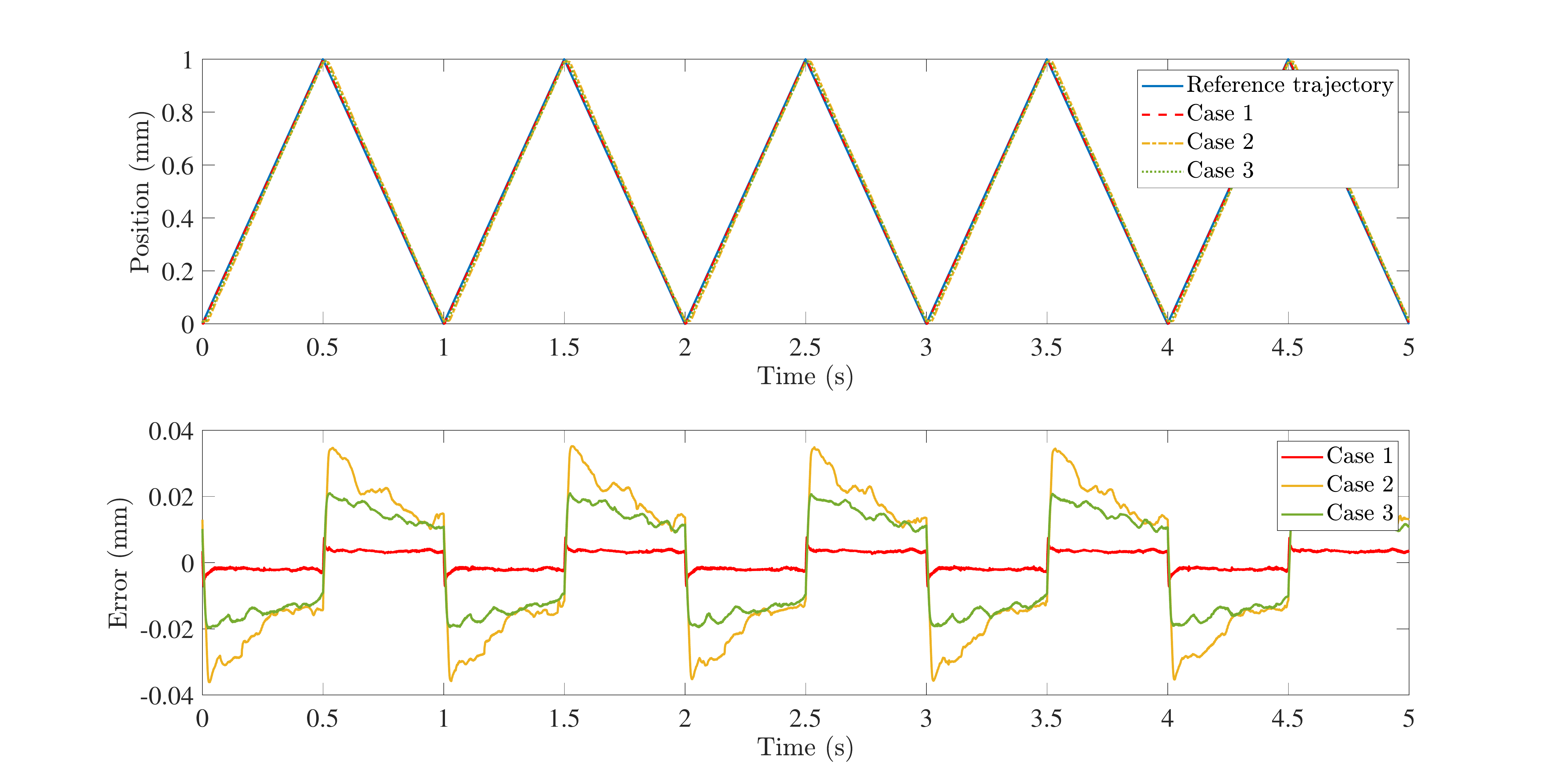}
	\caption{Trajectory tracking results at triangular wave using different controllers (\textbf{Top.} Position tracking performance; \textbf{Bottom.} Position tracking errors).\label{x_t_1Hz}}
\end{figure}
\begin{figure}[!b]
	\centering\vspace{-0.5cm}
	\includegraphics[scale=0.18,trim=40 10 50 0,clip]{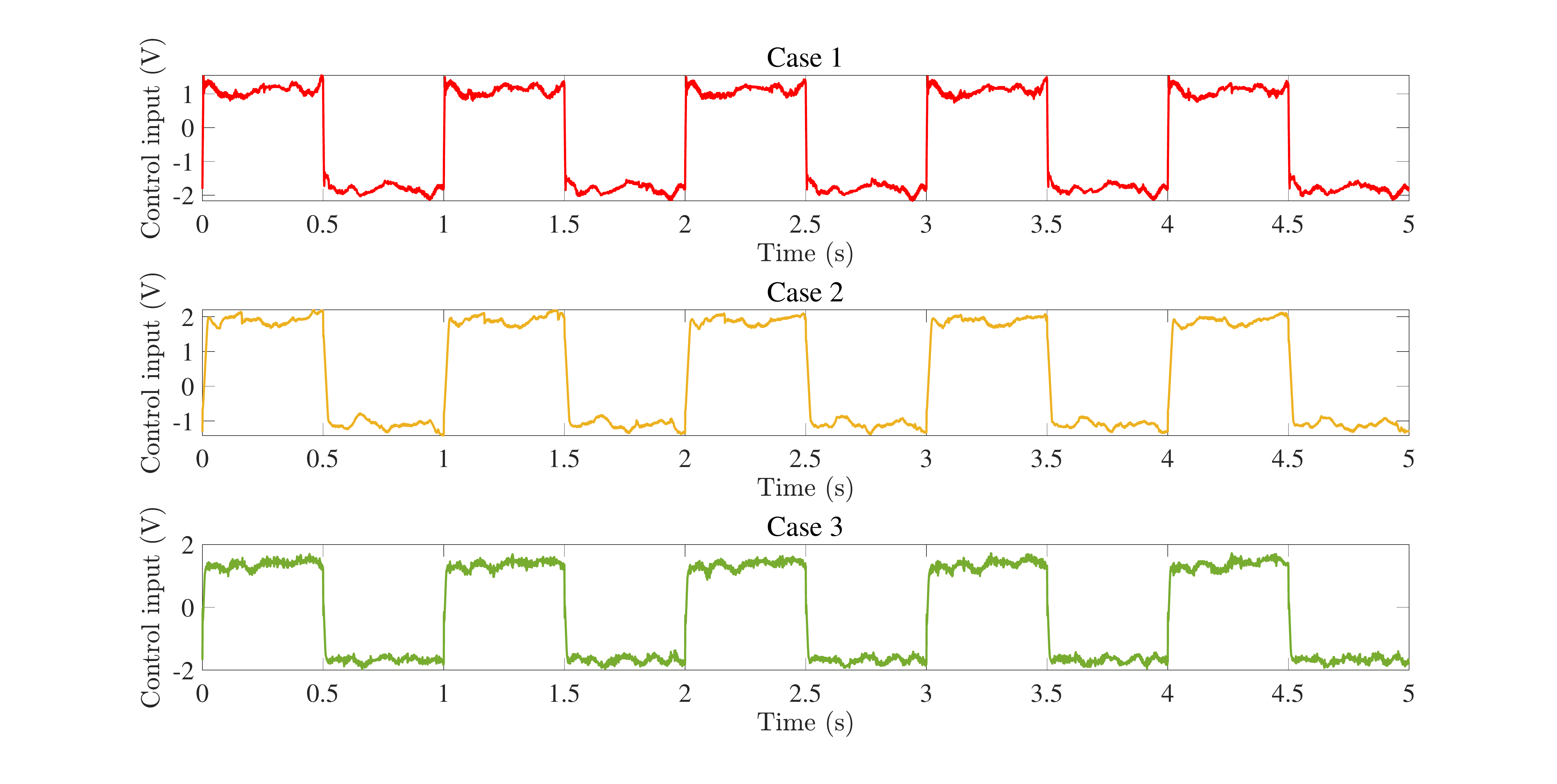}
	\caption{Comparison of control inputs at triangular wave among different controllers.\label{out_t_1Hz}}
\end{figure}

\cref{x_t_1Hz} shows the trajectory tracking results and errors at the triangular wave using different controllers. \cref{out_t_1Hz}  shows the control inputs at the triangular wave using different controllers.  It can be observed from these two figures that at the changing point, the designed controller achieves the fastest error convergence rate, and the peak value of its control input is the smallest.

\cref{error2_table} shows the comparisons of the MAEs and RMSEs among all these cases, where Case 5 is the results of our previous work presented in \cite{Yik2019motion} (a controller consists of a robust neural network and an extended state observer for the same experimental platform). Compared with the previous work, the MAE and RMSE are reduced by 57.79\% and 70.19\%, respectively. These results show that the designed controller can achieve higher control accuracy by consuming less energy to track discontinuous waves.
\begin{table}[!t]
	\centering
	\caption{Tracking errors of different cases at triangular wave.\label{error2_table}}
	\newsavebox{\secondtablebox}
	\begin{lrbox}{\secondtablebox}
		\begin{tabular}{lcccc}
			\hline
			\hline
			\tabincell{c}{MAE and \\ RMSE}& \tabincell{l}{Case 1} & \tabincell{l}{Case 2} & \tabincell{l}{Case 3} & \tabincell{l}{Case 5}\\
			\hline
			\tabincell{l}{MAE (mm)}& 0.0084& 0.0352 &0.0210  & 0.0199\\
			\tabincell{l}{RMSE (mm)} & 0.0031& 0.0216 & 0.0151 & 0.0104\\
			\hline
			\hline
		\end{tabular}
	\end{lrbox}
	\scalebox{1}{\usebox{\secondtablebox}}\vspace{-0.5cm}
\end{table}

\section{Conclusions}
This paper presents and details the design of an AFOSMC and a neural network compensator to ensure that a USM can run smoothly and accurately. In our work, the fraction-order calculus is introduced into the controller design to lower the chattering and improve the tracking accuracy of SMC. Here, the short memory principle is employed to solve the numerical solution of the fractional derivative during the implementation of the designed controller. Furthermore, the compensator is designed to compensate for the residual errors caused by the short memory principle. Finally, experiments verify the effectiveness of the designed controller, and the results show clearly that the designed controller can lower the chattering and achieve much better tracking performance.

%\section*{Acknowledgments}
%\indent\indent 
%\bibliographystyle{unsrt}
%\addtolength{\textheight}{-16.5cm}
\bibliographystyle{IEEEtran}
\bibliography{mybibfile}

\end{document}